\documentclass[10pt,journal,compsoc]{IEEEtran}

\ifCLASSOPTIONcompsoc
  \usepackage[nocompress]{cite}
\else
  \usepackage{cite}
\fi
\usepackage{algorithm}
\usepackage{algorithmic}
\usepackage{amsmath,amsfonts,bm,array,makecell,amsthm}
\usepackage{stmaryrd,graphicx,amsthm}
\usepackage{color}
\newtheorem{definition}{Definition}
\newtheorem{lemma}{Lemma}
\newtheorem{theorem}{Theorem}

\ifCLASSINFOpdf
\else
\fi

\begin{document}

\title{Space-Invariant Projection in \\ Streaming Network Embedding}

\author{Yanwen~Zhang,
        Huiwen~Wang,
        and~Jichang~Zhao
\IEEEcompsocitemizethanks{\IEEEcompsocthanksitem Yanwen Zhang is with the School of Economics and Management, Beihang University, Beijing, China and also with the Shen Yuan Honors College, Beihang University, Beijing 100191, China. \protect\\
E-mail:zhangyanwen@buaa.edu.cn

\IEEEcompsocthanksitem Huiwen Wang is with the School of Economics and Management, Beihang University, Beijing, China, and also with Key Laboratory of Complex System Analysis, Management and Decision (Beihang University), Ministry of Education.\protect\\
E-mail: wanghw@vip.sina.com

\IEEEcompsocthanksitem Jichang Zhao is with the School of Economics and Management, Beihang University, Beijing, China and also with Beijing Key Laboratory of Emergence Support Simulation Technologies for City Operations, Beijing, China.\protect\\
E-mail: jichang@buaa.edu.cn}
\thanks{Manuscript received *** **, 2023; revised *** **, 2023.\protect\\
(Corresponding author: Jichang Zhao)}}

\markboth{Journal of \LaTeX\ Class Files,~Vol.~**, No.~*, ***~2023}%
{Shell \MakeLowercase{\textit{Zhand et al.}}: Space-Invariant Projection in Streaming Network Embedding}

\IEEEtitleabstractindextext{%
\begin{abstract}
Newly arriving nodes in dynamics networks would gradually make the node embedding space drifted and the retraining of node embedding and downstream models indispensable. An exact threshold size of these new nodes, below which the node embedding space will be predicatively maintained, however, is rarely considered in either theory or experiment. From the view of matrix perturbation theory, a threshold of the maximum number of new nodes that keep the node embedding space approximately equivalent is analytically provided and empirically validated. It is therefore theoretically guaranteed that as the size of newly arriving nodes is below this threshold, embeddings of these new nodes can be quickly derived from embeddings of original nodes. A generation framework, Space-Invariant Projection (SIP), is accordingly proposed to enables arbitrary static MF-based embedding schemes to embed new nodes in dynamics networks fast. The time complexity of SIP is linear with the network size. By combining SIP with four state-of-the-art MF-based schemes, we show that SIP exhibits not only wide adaptability but also strong empirical performance in terms of efficiency and efficacy on the node classification task in three real datasets.

\end{abstract}

\begin{IEEEkeywords}
Dynamic network embedding, node embedding, matrix factorization, matrix perturbation theory, node classification.
\end{IEEEkeywords}}

\maketitle

\IEEEdisplaynontitleabstractindextext

%
\IEEEpeerreviewmaketitle

\section{Introduction}
\label{sec:introduction}

\IEEEPARstart{N}{etwork} embedding, or network representation learning, has gathered much research attention due to its superior performance in downstream tasks, such as node classification \cite{xu2019network,liu2021relative}, link prediction \cite{kazemi2018simple, rossi2021knowledge} and clustering \cite{han2019structured, guo2022graph}. The very goal of network embedding is to construct a low dimensional latent space that is able to capture high order proximities of the original graph. Two mainstream solutions of this task are matrix factorization (MF)-based \cite{liu2019general} and neural network (NN)-based methods \cite{zhou2022network}. From the very recent, the need for dynamic network embedding (DNE) in real life becomes noticeable, thus a series of methodologies proposed to solve this issue \cite{xue2022dynamic, barros2021survey}. As for the scene with new vertices arriving, some inductive methodologies based on graph neural networks were also developed \cite{hamilton2017inductive,liu2018semi,chu2019inductive,liu2019real,perini2022learning}.\par

Even though the effect of NN-based embedding is superior \cite{zhou2022network}, the importance of MF-based methods keeps irreplaceable. Firstly, the strategy of using sampling  to accomplish node proximity preservation makes NN-based schemes hardly utilize the whole network information, and also unstable due to different instantiations\cite{wang2020towards}. Secondly, many NN-based network embedding schemes have been proved to be uniformly correlated with the explicit matrix factorization \cite{qiu2018network, liu2019general, peng2022svd, agibetov2023neural}, and the development of both scientific computing and sparse matrix factorization accelerate and guarantee the feasibility of matrix factorization. Thirdly, MF-based methodologies are easily to be explained, due to the full development of theories concerning matrix and its perturbation.\par



 As for the DNE issue with new nodes adding, the drawback of NN-based solutions is also obvious. When extra features are needed to retrain models, the time complexity will be high, and is not affordable for streaming setting where quick embeddings are inherently necessary \cite{qi2020discriminative}. Thus a MF-based methodology with higher efficiency is in need for this scene. Even though some classic MF-based DNE methodologies mentioned that the incremental update using the change of adjacent matrix, denoting as $\Delta \bm{A}$, to update embedding, is able to be extended to generating new embedding by asserting zeros columns and lines in advance \cite{li2017attributed, wang2020dynamic, zhu2018high}, two inevitable drawbacks are involved in this strategy. At first, once it comes to using $\Delta \bm{A}$ to measure the change of networks, the number of new come vertices, $m$, should be determined in advance. Then, the threshold of accumulated error for incremental update that trigger restart is a hyper-parameter to determine \cite{zhang2018timers}, making the strategy experiment driven instead of data driven.

\begin{figure}[htbp]
    \centering
    \includegraphics[width=\linewidth]{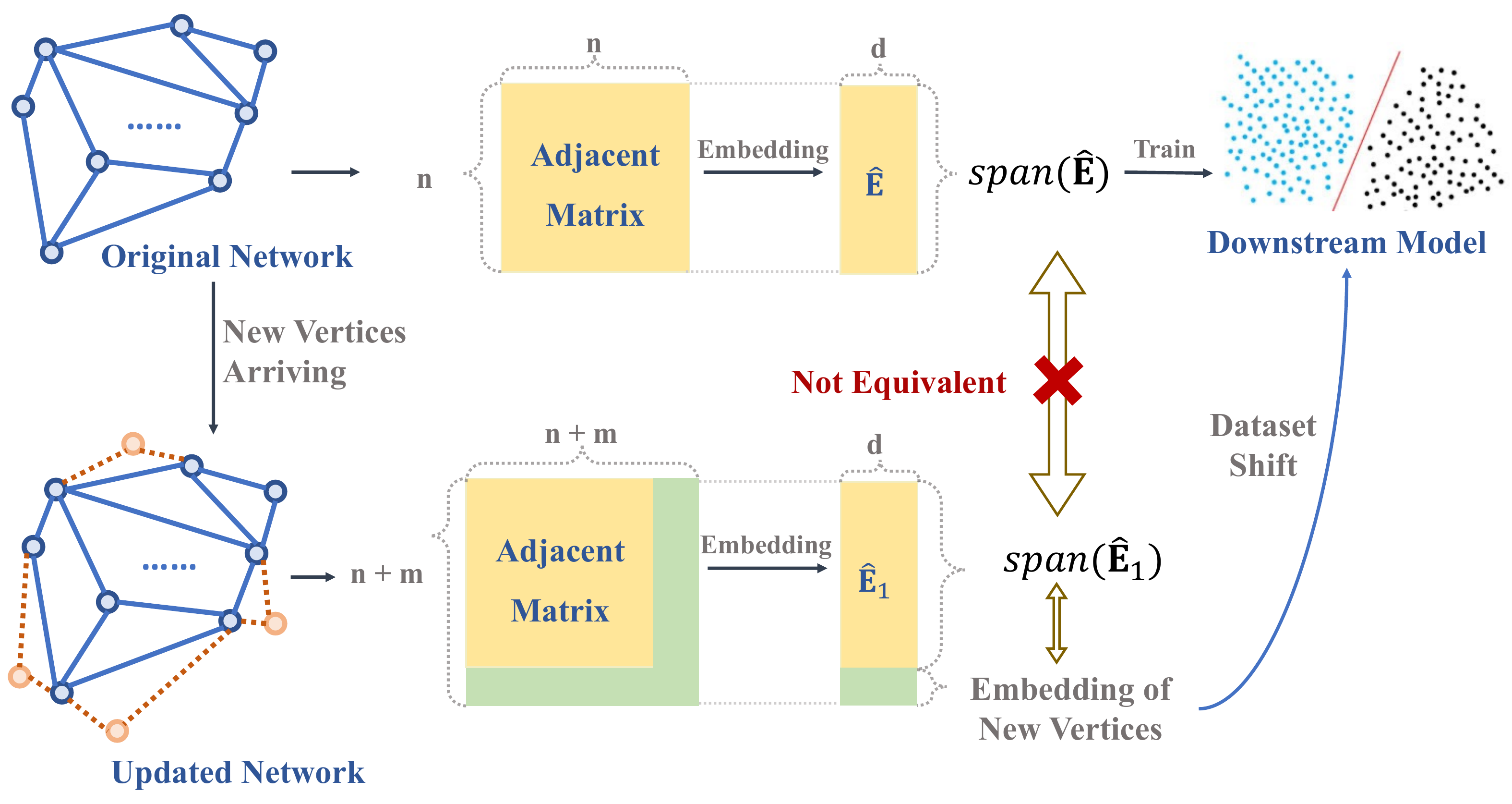}
    \caption{The illustration of embedding space drift and its consequence concerning downstream models.}
    \label{intro_fig}
\end{figure}

 In order to fix problems above, we should thoroughly understand the change of embedding space caused by new vertices arriving in advance. As illustrated in Fig.\ref{intro_fig}, let $\hat{\bm{E}}$ be the embedding of a graph with $n$ vertices, and $\hat{\bm{E}}_1$ is the embedding of corresponding $n$ vertices after $m$ more nodes arrived and representations regained. The ground truth embedding space of the updated network is acquired through retraining. With the increase of $m$, spaces spanned by columns of $\hat{\bm{E}}$ and $\hat{\bm{E}}_1$, denoted respectively as $span\{\hat{\bm{E}}\}$ and $span\{\hat{\bm{E}}_1\}$, will eventually become hardly equivalent, which cause the figures of each embedding, i.e. each corresponding row of $\bm{E}$ and $\hat{\bm{E}}_1$, become gradually differentiated. The drift of space inevitably deteriorate the performance of the downstream model, thus a problem similar to dataset shift \cite{quinonero2008dataset} will occur, even though the shifts here are resulted from the embedding vectors rather than the sample distribution. \par 

Thus a natural question, which has been rarely explored theoretically, is that what size of these newly arrived nodes will eventually undermine the original embedding space and make the restart for both embedding and the downstream model necessary. A theoretical solution of this threshold size $m$ would be critical in real applications, since it will help avoid the unnecessary embedding retraining. It also worthy noting that $m$ here is totally determined by the network structure, the hyper-parameter error bound is no longer an input.

Thereafter, as the size of newly arrived nodes is smaller than $m$, space drift is acceptable and no needs to retain the embedding model. In line with this, we propose a novel projection embedding model inspired by theoretical analysis to handle the streaming network embedding issue containing newly introduced nodes. Actually, for new nodes entering the network without additional information, if they are interacted with existing nodes, we can reasonably generalize the characteristics learned from existing vertices to them, just as the common solution to the ubiquitous cold-start problem, e.g. utilizeing links in social networks to obtain user preferences \cite{zhang2020joint}. 
The time complexity of proposed solution is linearly to the node size. \par

To verify the advantages of our framework, we combine it with four outstanding MF-based static schemes, and conduct node classification experiments for new arrived nodes on three real-world networks respectively. The result demonstrates the outperformance of our framework in efficacy, efficiency and adaptation for different schemes. 

The main contributions of this paper are summarized as follows:\par

\begin{itemize}
\item For the MF-based dynamic network embedding problem with nodes arriving, we propose a criteria for judging when to restart, which is able to calculate the threshold of the number of new nodes that above which both the embedding and downstream model need to be retrained.
\item For new nodes, different from the idea of incremental update proposed by the previous methods to obtain embedding, we propose a fast projection framework, SIP, which can be combined with different MF-based methods to construct embedding. The time complexity is linear to the network size.
\item The effecacy and efficiency of both threshold theory and projection embedding are validate on four MF-based static embedding schemes using three real-world networks.

\end{itemize}
The rest of the paper is organized as follows. In section 2, literature review is presented. Section 3 defines the problem, and in section 4, all theoretical formulations including the deduction of $m$ and embedding generation framework are detailedly elaborated, with four instantions of SIP presented. We then focus on node classification task for new arriving nodes in section 5, and evaluates our performance thoroughly. The last section is about conclusion and discussion.

\section{Related Work}

\subsection{Static Network Embedding}

To learn the topology information and preserve the node proximity of static networks, traditional dimension reduction approaches using matrix factorizaiton, like Graph Laplacian Eigenmaps (LE)\cite{tenenbaum2000global,belkin2003laplacian} and non-negative factorization (NMF) \cite{sun2014alternating, ye2018deep} are adopted as classic embedding methods. Recetly, GLEE \cite{torres2020glee} innovatively promoted the LE to capture geometric similarities, which outperformed LE in many aspects. Factorization on a high-order proximity matrix is also a common idea for node embedding. HOPE \cite{ou2016asymmetric} leveraged a generalized Singular Value Decomposition (SVD) to factorize proximity-preserved matrices constructed based on different metrics. AROPE \cite{zhang2018arbitrary} proposed an efficient scheme to preserve arbitrary-order proximities based on the real requirment. GraRep \cite{cao2015grarep} factorized the top k-order transition matrices and concatenating the vectors together, which effectively captured the global information of nodes, and outperformed other state-of-the-art methods in specific tasks.\par 

Moreover, DeepWalk \cite{perozzi2014deepwalk} firstly adopted SkipGram-alike model \cite{mikolov2013distributed}, treating nodes as words and random walks as sentences, the output of which were node embeddings. Node2vec \cite{grover2016node2vec} introduced biased random walk with hyper parameters, and role2vec \cite{ahmed2020role} introduced attributed random walk to extend this work. These two methods extended DeepWalk to adopt different scenarios. While in reality, when the edge weights have a high variance, the gradient explosion or disappear will occur, thus LINE \cite{tang2015line} was proposed to address this using elaborately designed objective functions that preserve the first-order or second-order proximity. Based on the pioneering work demonstrating the equivalence between skip-gram model and matrix factorization \cite{levy2014neural}, NetMF \cite{qiu2018network} created an intrinsic connection between DeepWalk, LINE, PTE \cite{tang2015pte} (an extension of second-order LINE in heterogeneous text network), and node2vec, thus these time-consuming skip-gram based network embedding models are unified into matrix factorization.

In order to conquer the high embedding time complexity on large scale graphs, many upgrades were invented recently. RandNE \cite{zhang2018billion} implemented Gaussian random projection to optimize the matrix factorization object functions, which is parallelly computable thus time-saving. NodeSketch \cite{yang2019nodesketch} preserved node proximity using recursive sketching, which generated node embedding based on the self-loop-augmented adjacent matrix.\par

\subsection{Dynamic Network Embedding}

Due to the dynamics of most real-world networks, to embed networks evolving over time becomes more impelling.\par 

For NN-based methods, at first, some attempts were paid to extend existed embedding methodologies to the dynamic setting \cite{mahdavi2018dynnode2vec, peng2020dynamic, sajjad2019efficient, barracchia2022lp, zhou2018dynamic, trivedi2019dyrep, goyal2018dyngem}. 
Moreover, many methods attempted to accelerate the embedding process \cite{hou2020glodyne, zhang2021dynamic, xie2021learning,du2018dynamic, li2022robust}. For example, Du $et \ al.$ \cite{du2018dynamic} only updated the most affected nodes when the graph changes, and AdaNet \cite{li2022robust} adaptively determined the nodes to be update using pre-learned model with an intention to consider both the new information and time smoothness. To preserve the embedding stability for the newly arrived vertices, a series of inductive methods were proposed. GraphSAGE \cite{hamilton2017inductive} incorporated node features and Liu $et \ al.$ \cite{liu2018semi} adopted edge features to generate new embeddings. MINE \cite{chu2019inductive} promoted these methods into heterogeneous networks. For streaming graphs, Wang $et \ al.$\cite{wang2020streaming} trained Graph Neural Networks (GNN) that are able to preserve both new patterns and existed patterns simultaneously and Perini $et \ al.$ \cite{perini2022learning} extended GNN to let it being capable of processing streaming data. However, NN-based methodologies inevitably need training, thus the time needed for convergence is always consuming.

For MF-based methods, incremental update is the earliest idea. DANE \cite{li2017attributed} proposed a scheme to capture both the network and attribute changes, while dyHNE \cite{wang2020dynamic} introduced meta-path based proximities and updated embedding through eigenvalue perturbation. DHPE \cite{zhu2018high} used generalised SVD and its eigen perturbation to generate and update the embedding. NPDNE \cite{yu2020node} adopted NetMF to generate embedding and updated them. All methods mentioned above used matrix perturbation theory to update, instead of retraining, eigenvectors and eigenvalues, while errors will be inevitably accumulated. To fix this, TIMERS \cite{zhang2018timers} theoretically constructed an analytical guidance for the proper timing to restart embedding training from the error perspective. Moreover, using temporal smoothing for optimization is also a solution \cite{zhu2016scalable, ferreira2019modeling}, such as a online approximation method developed by Liu $et \ al.$ \cite{liu2019real}, which is motivated by Weighted Independent Cascade Model. Despite all these development, generating embedding for new nodes is still a tough problem for MF-based methods, and only a few recent works have attempted to solve it \cite{levin2018out, liu2019real}. Within this scene, a theoretical understanding of how the embedding space drift after dynamic change is absent, even though it is a crucial concern, and the result of which is a probable guide for timing chosen to retrain models. Therefore, in this paper, we will try to accomplish this theoretical exploration, and propose a scheme being able to generate new embeddings based on the theory.\par


\section{Problem Setting}
In this section, we formally define the problems that will be addressed in this paper.\par
\vspace{0.5em}
\noindent\textbf{Streaming graph embedding.} Streaming graph is a special kind of dynamic graph where the time is treated as continuous instead of discrete and new nodes are added. Graph streams are commonly seen in real-world problems. For example, in the citation networks, new articles as new nodes and new citations as new links are added in order. Our streaming graph embedding problem is defined as follows.

\begin{definition}[Graph Embedding]
    Given a time $t \in \mathbb{N}$, a graph at time $t$ is denoted as $\mathcal{G}_t = (\mathcal{V}_t, \ \mathcal{E}_t, \ \bm{A}_t)$, where $\mathcal{V}_t$ and $\mathcal{E}_t$ serves as the node set and edge set respectively, and $\bm{A}_t$ denotes the adjacent matrix of this graph at time $t$. The goal of graph embedding is to learn a function $\mathcal{V}_t \ \rightarrow \mathbb{R}^{d}$ that maps each vertice into a $d-$dimensional vector $\bm{b}_{v}$, stacking which vertically forms the embedding matrix $\bm{B}$ with $\vert\mathcal{V}_t \vert$ rows and $d$ columns.
\end{definition}

\begin{definition}[Streaming Graph Embedding]
    Denote $t_0$ as the beginning time with $\mathcal{G}_{t_0} = (\mathcal{V}_{t_0}, \ \mathcal{E}_{t_0}, \ \bm{A}_{t_0})$. At any time $t = t_0+ \Delta t$ with $\mathcal{G}_{t_0+ \Delta t} = (\mathcal{V}_{t_0+ \Delta t}, \ \mathcal{E}_{t_0+ \Delta t}, \ \bm{A}_{t_0+ \Delta t})$ where new vertices $\Delta \mathcal{V}$ and associated edges $\Delta \mathcal{E}$ are added, the goal of streaming graph embedding is to generate $\{\bm{e}_v\}_{v \in \Delta\mathcal{V}}$ and update $\{\bm{e}_v\}_{v \in \mathcal{V}_{t_0}}$.
\end{definition}

\vspace{0.5em}
\noindent\textbf{Node classification.} For dynamic network embedding, a classic task is node classification, that utilizing a supervised model to capture the labels information of nodes. This task is aim to solve two sub-problems: adds unknown labels about known nodes in the network, or identifies the labels of newly arrived vertices. The latter question is the target of this paper, which is significant for many real-world tasks, like identifying the topic of specific papers in the citation network, and predicting users' professions in social networks.


\section{Proposed Method: SIP}

In this section, we theoretically elaborate the space-invariant projection framework for streaming network embedding, or SIP for short. At first, we deduces the upper bound of the size of newly arrived nodes, i.e., $m$, more than which the restart of both embedding and downstream models’ training become necessary. Thereafter, the rationality of using projection to generate new embedding is clarified. 
Moreover, since SIP is a general framework being able to extend a range of static network embedding methods, we take Laplacian Eigenmaps (LE), NetMF \cite{qiu2018network}, AROPE \cite{zhang2018arbitrary} and GraRep \cite{cao2015grarep} as examples to construct dynamic methodologies. Noting that in this paper we mainly focus on the undirected networks, whose adjacent matrix are symmetric.

\subsection{Discussion of the restart threshold}
\label{threshold}


Since different MF-based methodologies construct different matrices to decompose, such as polynomial function of the adjacent matrix with different order \cite{yang2017fast}, we collectively denote them as $\bm{M}$, with their eigen-decomposition as $\bm{M} = \bm{U \Sigma U}^\text{T}$, which is also the SVD result as we focus on the situation that $\bm{M}$ is symmetric. The embedding matrix is always chosen as $\bm{U}_d$, the first $d$ column of $\bm{U}$, or the linear scaling associated with it such as $\bm{U}\sqrt{\bm{\Sigma}_d}$. Under this simplification, the drift of network embedding converts to the drift of eigenvectors, which span an invariant subspace of $\bm{M}$. Therefore, the analysis below will focus on the change of $\bm{U}$.

 Assume that there are $n$ nodes at the beginning time $t_0$, when the static network embedding methodology is applied to corresponding $\bm{M}_0$ to obtain the original embedding. Then $m$ more nodes arrived, with the object matrix becomes $\Tilde{\bm{M}}$. To connect $\bm{M}_0$ and $\Tilde{\bm{M}}$, here we denote
\begin{equation}
\bm{M} = \textbf{$\left[
    \begin{matrix}
    \bm{M}_0 && \bm{O}_{n\times m}\\
    \bm{O}_{m \times n} && \bm{O}_{m \times m}\\
    \end{matrix}
    \right]
    $},
\end{equation}
where $\bm{O}_{p \times q}$ is the zero matrix with $p$ rows and $q$ columns. For the convenience of following descriptions, let $\bm{X} = \bm{U}_d$ as the first $d$ eigenvectors of $\bm{M}$ and $\bm{Y}$ as the matrix composed of the remaining $(n+m-d)$ eigenvectors. Without loss of generality, here we assume $rank(\bm{M}) = n$. Then $\bm{X}$ and $\bm{Y}$ satisfies
$$
\bm{MX} = \bm{X} \bm{\Sigma}_d, \ \bm{\Sigma}_d
 = diag(\sigma_1, \ \sigma_2, \ ..., \ \sigma_d),
 $$
 $$
 \bm{MY} = \bm{Y} \hat{\bm{\Sigma}}_{n-d}, \ \hat{\bm{\Sigma}}_{n-d} = diag(\sigma_{d+1}, \ ..., \ \sigma_{n}, \ 0, \ ... , \ 0),
 $$
where $\sigma_i$ is the $i^{\text{th}}$ largest eigenvalue of $\bm{M}_0$. In addition, these two matrices can be written as partitioned matrices. To be specific,
$$
\bm{X} = \textbf{$\left[
    \begin{matrix}
    \bm{U}^{(1)}\\
    \bm{O}_{m \times d}\\
    \end{matrix}
    \right]
    $}, \ \text{and} \
\bm{Y} = \textbf{$\left[
    \begin{matrix}
    \bm{U}^{(2)} && \bm{O}_{n \times m}\\
    \bm{O}_{m \times (n-d)} && \bm{U}^{(3)}\\
    \end{matrix}
    \right]
    $},
$$
where $\bm{U}^{(1)} \in \mathbb{R}^{n \times d}$ is the first $d$ eigenvectors of $\bm{M}_0$, $\bm{U}^{(2)} \in \mathbb{R}^{n \times (n-d)}$ is the remaining $(n-d)$ eigenvectors and $\bm{U}^{(3)} \in \mathbb{R}^{m \times m}$ is a column orthonoraml matrix just to make columns of $\bm{Y}$ span a $(n+m-d)$ dimensional space. Then $\Tilde{\bm{M}}$ can be represented as
\begin{equation}
\Tilde{\bm{M}} = \bm{M} + \textbf{$\left[
    \begin{matrix}
    \Delta \bm{M} && \hat{\bm{E}}^{(1)}\\
    \hat{\bm{E}}^{(1)\text{T}} && \hat{\bm{E}}^{(2)}\\
    \end{matrix}
    \right]
    $} = \bm{M} + \bm{E}.
\label{2}
\end{equation}
$\Delta \bm{M} \in \mathbb{R}^{n \times n}$ represents the impact of network changes on the initial $n$ nodes, including the adding and deleting of nodes and edges. $\hat{\bm{E}}^{(1)} \in \mathbb{R}^{(n-m) \times m}$ is the network between $m$ new coming nodes and $n$ existed nodes, while $\hat{\bm{E}}^{(2)} \in \mathbb{R}^{m \times m}$ consists of mutual relations between $m$ new nodes. \par

To analytically explore networks with nodes add or removal, matrix perturbation theory was widely adopted \cite{chen2015node}, which is also the method to be used in this paper. Denoting $\rho_{\bm{S}}$ as the maximum singular value of matrix $\bm{S}$, and
\begin{equation}
    \text{sep}(\bm{\Sigma}_d, \ \hat{\bm{\Sigma}}_{n-d}) = \inf \limits_{\|\bm{P}\|=1} \|\bm{P\Sigma}_d - \hat{\bm{\Sigma}}_{n-d}\bm{P}\|>0
\end{equation}
then we can present the following theorem and a extended lemma .

\begin{theorem}
\label{subspace}
 If $\bm{X}$ is an invariant subspace of $\bm{M}$, and $(2\rho_{\hat{\bm{E}}^{(1)}} + \rho_{\hat{\bm{E}}^{(2)}} + \rho_{\Delta \bm{M}}) < \text{sep}(\bm{\Sigma}_d, \ \hat{\bm{\Sigma}}_{n-d})$, then, there is a unique matrix $\bm{P}$ such that the columns of 
 \begin{equation}
 \label{theorem_eq}
     \Tilde{\bm{X}} = \bm{(X+YP)(I + P^{\text{T}}P)}
 \end{equation} 
 form the orthonormal basis for the invariant subspace of $\Tilde{\bm{M}}$.
\end{theorem}

\begin{proof}
Since $\bm{MX=X\Sigma}_d$, $span(\bm{X})$ is an invariant subspace of $\bm{M}$. Also, due to the columns of $\bm{Y}$ span $span(\bm{X})^{\perp}$, then
$$
\bm{(X \ Y)^{\text{T}} M (X \ Y)} = \left[
    \begin{matrix}
    \bm{\Sigma}_d && \bm{O}_{d \times (n+m-d)}\\
    \bm{O}_{(n+m-d) \times d} && \hat{\bm{\Sigma}}_{n-d}\\
    \end{matrix}
    \right].
$$
After the perturbation, in order to calculate $\bm{(X \ Y)^{\text{T}} E (X \ Y)}$, let
\begin{align*}
    \bm{(X \ Y)}
    = & \left[
    \begin{matrix}
    \bm{U}^{(1)} & \bm{U}^{(2)} & \bm{O}_{n \times m}\\
    \bm{O}_{m \times d} & \bm{O}_{m \times (n-d)} & \bm{U}^{(3)}
    \end{matrix}
    \right] \\
    = & \left[
    \begin{matrix}
    \bm{U}_{n} & \bm{O}_{n \times m} \\
    \bm{O}_{m \times n} & \hat{\bm{U}}_{m \times m}
    \end{matrix}
    \right],
\end{align*}
where $\bm{U}_n$ is the horizontal stack of $\bm{U}^{(1)}$ and $\bm{U}^{(2)}$, then
\begin{align*}
\bm{(X \ Y)^{\text{T}} E (X \ Y)} = & \left[
    \begin{matrix}
    \bm{U}_{n}\Delta\bm{M}\bm{U}_{n}^{\text{T}} && \bm{U}_{n}^{\text{T}}\hat{\bm{E}}^{(1)}\hat{\bm{U}}\\    \hat{\bm{U}}^{\text{T}}\hat{\bm{E}}^{(1)\text{T}}\bm{U}_{n} && \hat{\bm{U}}^{\text{T}}\hat{\bm{E}}^{(2)}\hat{\bm{U}}^{(1)}\\
    \end{matrix}
    \right]\\
    = & \left[
    \begin{matrix}
    \bm{E}_{11} && \bm{E}_{12}\\    
    \bm{E}_{21} && \bm{E}_{22}\\
    \end{matrix}
    \right].
\end{align*}
According to the perturbation theory of invariant subspaces (\cite{stewart1990matrix}, p236), let $\| \cdot \|$ represent a consistent family of unitarily invariant norms and set
\begin{align*}
    \gamma = & \| \bm{E}_{21} \| = \| \bm{E}^{(1)} \| \\
    \eta = & \|\bm{E}_{12} \| = \| \bm{E}^{(1)} \| \\
    \delta = &\text{sep}(\bm{\Sigma}_d, \ \hat{\bm{\Sigma}}_{n-d}) -  \| \bm{E}_{11} \| - \| \bm{E}_{22} \| \\
    = & \text{sep}(\bm{\Sigma}_d, \ \hat{\bm{\Sigma}}_{n-d}) -  \| \Delta \bm{M} \| - \| \hat{\bm{E}}^{2} \|,
\end{align*}
where $\text{sep}(\bm{\Sigma}_d, \ \hat{\bm{\Sigma}}_{n-d}) = \inf \limits_{\|\bm{P}\|=1} \|\bm{P\Sigma}_d - \hat{\bm{\Sigma}}_{n-d}\bm{P}\|>0$.
Then, if $\delta > 0$ and $\delta^2 > 4*\eta \gamma$, the unique $\bm{P}$ exists that satisfies requirement in the theorem. Since here $\eta = \gamma > 0$, $\delta^2 > 4*\eta \gamma$ converts to $\delta > 2*\eta$ which naturally satisfies $\delta > 0$.\par

Let $\|\cdot \|$ to be the spectral norm, then $\|\hat{\bm{E}}^{(1)}\| = \rho_{\hat{\bm{E}}^{(1)}}$, $\|\hat{\bm{E}}^{(2)}\| = \rho_{\hat{\bm{E}}^{(2)}}$ and $\|\Delta \bm{M}\| = \rho_{\Delta \bm{M}}$, thus equation (\ref{requirement}) becomes
\begin{equation}
    \rho_{\Delta \bm{M}} + 2\rho_{\hat{\bm{E}}^{(1)}} + \rho_{\hat{\bm{E}}^{(2)}} < \text{sep}(\bm{\Sigma}_d, \ \hat{\bm{\Sigma}}_{n-d}).
\end{equation}
Therefore, the conclusion is established.
\end{proof}

Many conclusions are revealed by deduction above. At first, according to the definition of invariant subspace, if $\bm{X}$ is still the basis of invariant subspace of $\Tilde{\bm{M}}$ after perturbation, namely the subspace drift does not happen, all elements in $\bm{E}_{21} = \hat{\bm{U}}^{\text{T}}\hat{\bm{E}}^{(1)\text{T}}\bm{U}_{n}$ should equal to zero. While $\bm{U}_n$ and $\hat{\bm{U}}$ are orthonormal matrices, if and only if $\hat{\bm{E}}^{(1)}$ equals to a zero matrix can the requirement be satisfied, which means that the new arrived nodes have no interaction with previous vertices. However, practically, this is unrealistic for a network in real-world, and theoretically, under the circumstances, $\bm{X}$ is no longer the eigenvector matrix of the whole $\Tilde{\bm{M}}$, since the rank of space that $\bm{X}$ spans is at most $n$. More importantly, if $\hat{\bm{E}}^{(1)} = \bm{O}$ and we regard $\bm{X}$ as the embedding matrix, the embeddings of new vertices are all equal to zero vectors, which is obviously not in line with our demands. Therefore, the space drift is almost sure to happen, as long as there are new nodes arrived.\par

Then, we are supposed to figure out that ``how much'' drift will happen to the $d-rank$ subspace $\bm{X}$, and how to measure the space difference. According to Theorem \ref{subspace}, if the requirement is satisfied, there exists a $\Tilde{\bm{X}}$ that is constructed through rotation and shift of $\bm{X}$, while those acts that turning $\bm{X}$ to $\Tilde{\bm{X}}$ are all determined by $\bm{P}$. Again, according to the matrix perturbation theory (\cite{stewart1990matrix}, p232), the singular values of $\bm{P}$ are actually tangents of canonical angles between $span(\bm{X})$ and $span(\Tilde{\bm{X}})$, and more importantly,
\begin{equation}
    \|\bm{P}\| < 2\frac{\gamma}{\delta},
\end{equation}
thus to some extent, the existence of $\bm{P}$ makes the difference between the two subspaces measurable and limitable, since the tangent of the max canonical angle has a upper bound. While, since $d$, the rank of subspace $\bm{X}$ and $\Tilde{\bm{X}}$, is much smaller than $(n+m)$, the rank of $\Tilde{\bm{M}}$, thus the subspaces constructed by $\bm{X}$ and $\Tilde{\bm{X}}$ can be two completely different $d-$ rank subspaces in the $(n+m)$ dimensional manifold. Intuitively speaking, only when these two subspaces are similar enough, can the tangent of canonical angles be limited. This may also indicate that, under the existence of $\bm{P}$, the first a few columns of of $\bm{X}$ and $\Tilde{\bm{X}}$, namely their eigenvectors corresponding to the largest eigenvalues respectively, are highly correlated. This inference has important guiding significance if its correctness can be proved. However, the calculation of $\text{sep}(\bm{\Sigma}_d, \ \hat{\bm{\Sigma}}_{n-d})$ is difficult. In order to apply Theorem \ref{subspace}, a looser condition under the existence of $\bm{P}$ is proposed below.\par

\begin{lemma}
\label{lemma_requirment}
    When the matrix $\bm{P}$ satisfies Theorem \ref{subspace} exists, 
    \begin{equation}
        2\rho_{\hat{\bm{E}}^{(1)}} + \rho_{\hat{\bm{E}}^{(2)}} + \rho_{\Delta \bm{M}} < \sigma_1 - \sigma_{2}
    \label{judgement_equation}
    \end{equation}
    is established.
\end{lemma}

\begin{proof}
    Need to add that, 
\begin{equation}
    \text{sep}(\bm{S}_1, \ \bm{S}_2) \leq min|\mathcal{L}(\bm{S}_1) - \mathcal{L}(\bm{S}_2)|
\end{equation}
where $\mathcal{L}(\bm{S})$ is the set of eigenvalues of $\bm{S}$ (\cite{stewart1990matrix}, p233). Thus
\begin{equation}
    \text{sep}(\bm{\Sigma}_d, \ \hat{\bm{\Sigma}}_{n-d}) \leq (\sigma_d - \sigma_{d+1}) < \sigma_1 - \sigma_2
\end{equation}
Therefore the result converts to the requirement that
\begin{equation}
    \rho_{\Delta \bm{M}} + 2\rho_{\hat{\bm{E}}^{(1)}} + \rho_{\hat{\bm{E}}^{(2)}} < \sigma_1 - \sigma_2 .
\label{requirement}
\end{equation}
\end{proof}

Therefore, practically, after $m$ new nodes arrived, we can directly calculate whether equation \ref{judgement_equation} is satisfied to indicate whether the retraining is in need. The maximum number $m$ satisfies the equation is denoted as $m_0$.



\subsection{Empirical evidence of the restart threshold}
\label{empirical restart}

Then, we will use $m_0$ calculated in different datasets and different $n$ to validate the efficiency of theory above.

Three datasets are leveraged here, the Protein-Protein Interaction subgraph \cite{stark2010biogrid}, Blogcatalog and Flickr \cite{tang2009relational}. To make them more reasonable as time-varying networks, those not in the giant connected component were omitted in advance.\par 

Four MF-based embedding methods, Laplacian Eigenmaps (LE) \cite{belkin2001laplacian}, NetMF \cite{qiu2018network}, AROPE \cite{zhang2018arbitrary} and GraRep \cite{cao2015grarep} are employeed to construct target matrices, i.e. $\bm{M}$ mentioned above. To be specific, for LE, the target matrix is the normalized laplacian matrix, while for NetMF, the target matrix is the deepwalk matrix proposed in the paper. For AROPE, the target matrix is simply $\bm{S} = w_1 \bm{A} + w_2 \bm{A}^2 +... + w_q\bm{A}^q$ as the paper defined, with $q=3$ and $w_i = (0.01)^{i-1}$ chosen. Since in GraRep there are $k$ positive log probability matrices constructed through $(\bm{D}^{-1}\bm{A})^i$ for $1 \leq i \leq k$ to decompose, we just choose the $k^{th}$ order one as the target matrix, since it has the minimum $m_0$. Ideally, all constructed matrices are supposed to process $m_0$ under given $n$. However, for LE, whose $\bm{M}$ are normalized matrices related to the adjacent matrix, their absolute values are small, thus the singular values are small, which leads the unsatisfaction of Theorem \ref{subspace} even when $m=0$. Therefore, for LE, we choose adjacent matrix $\bm{A}$ as the target matrix.

After $\bm{M}$ with size $n$ is constructed, its first $128$ eigenvectors are calculated. Thereafter, $\bm{M}$ with size $n+m$ for different $m$ is also constructed and decomposed. We are supposed to figure out how the space spanned by the embedding of the first $n$ vertices drift after retraining, thus the correlation coefficient of the first $n$ elements in first $128$ eigenvectors are chosen as the metric here. That’s to say, a high coefficient means no significant data drift and the embedding spaces is therefore invariant, i.e., the retraining is unnecessary. The results are plotted in Fig.\ref{heatmap_fig} where the horizontal axis indicates the order of eigenvectors, and the vertical axis is the $m$ index, with the numerical result of maximum integer $m_0$ satisfies Theorem \ref{subspace} marked out. 


\begin{figure*}[htbp]
    \centering
    \includegraphics[width=\linewidth]{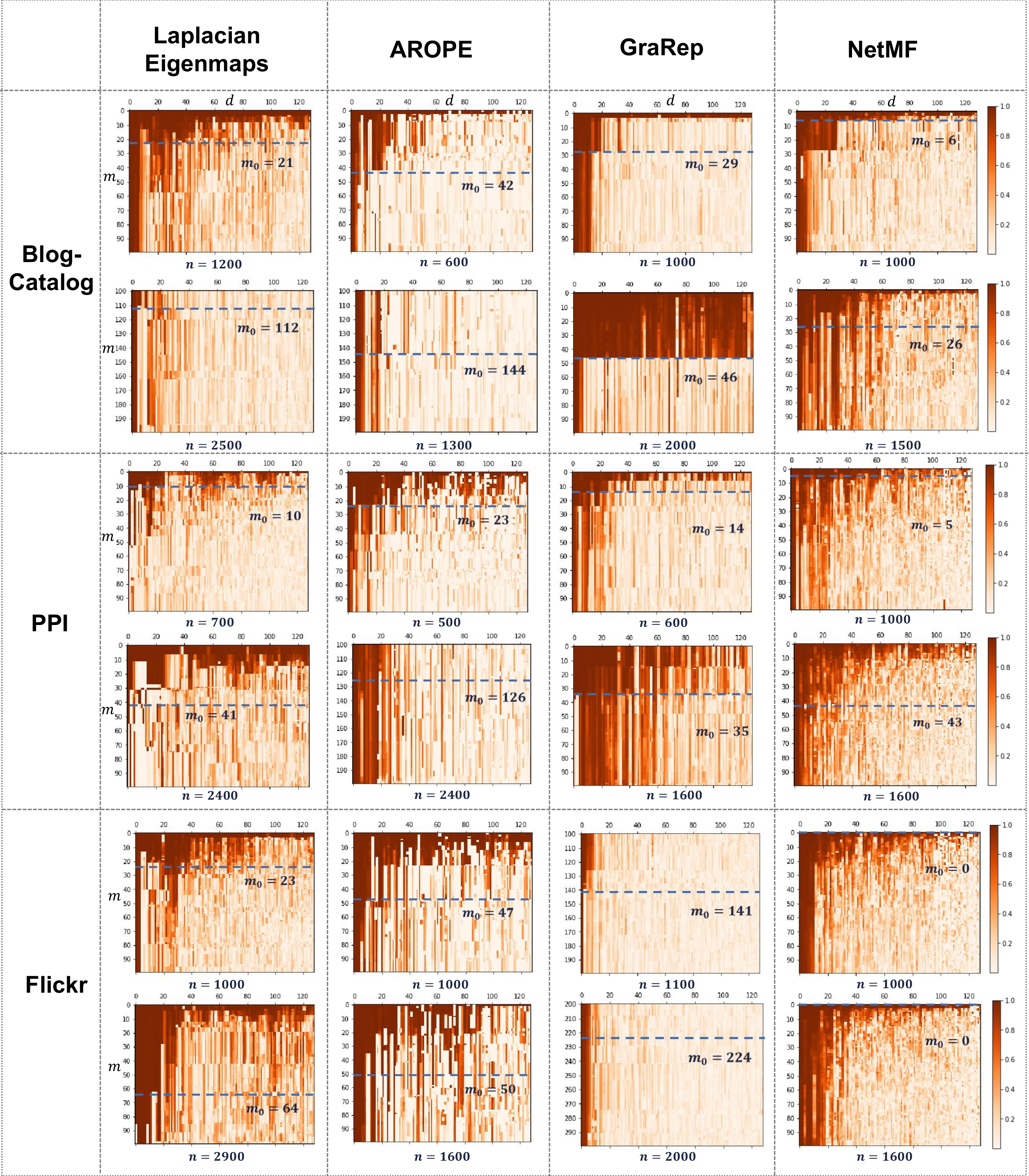}
    \caption{Empirical evidence of Theorem 1 through the use of three different real networks. We firstly construct the target matrix of different methodologies in different sizes to decompose, then calculate the correlation coefficient of the first $n$ elements of corresponding eigenvectors, between matrix with size $n$ and matrix with size $n+m$. $m_0$ is the theoretical upper bound of $m$ to ensure the existence of $\bm{P}$ in Theorem 1.}
    \label{heatmap_fig}
\end{figure*}

As can be seen in Fig.\ref{heatmap_fig}, at first, the inference that the correlation coefficient will come across a ``drop'' with the increase of $m$ is empirically proved, for nearly all matrices constructed by different methodologies and different datasets. Moreover, the guiding significance of Lemma \ref{lemma_requirment} for choosing $m_0$ is apparent, even though it is not the strict necessary and sufficient condition of the existence of $\bm{P}$. For example, for the target matrix constructed by GraRep using the Blogcatalog dataset, when $n = 1000$ and $n = 2000$, around the theoretically limitation $m_0$, a sudden drop of the correlation coefficient is observed for $d > 20$ and $d > 5$ respectively.
 Although in some cases, the correlation coefficient around the value of $m_0$ is not remarkably different, such as $\bm{M}$ constructed by AROPE in the PPI dataset when $2400$, there are still only the first few eigenvectors are highly correlated,  which is significantly lower than the common embedding dimension $64$ or $128$, thus the necessity of retraining the embedding when $m > m_0$ is also hold. Another case that needs to be elucidated is that $m_0 = 0$ is probably met for some circumstances. For example, matrices constructed by NetMF using the Flickr dataset always meet the situation $m_0=0$. It can be seen from the figure that this is caused by the fact that the high correlation coefficient of all $128$ eigenvectors only appearing in the first few $m$, while which is quickly reducing with the increase of $m$ except for the first few $d$. \par

\subsection{Embedding Generation through projection}

As discussed both in the section \ref{sec:introduction} and section \ref{threshold}, the purpose of figuring out the threshold of space shift is to guide the inevitably retraining of both embedding and the downstream model. Thus inversely, when the number of arrived nodes, $m$, is smaller than the threshold, it is possible and reasonable to both update, instead of retraining, the existed embedding and generate new embedding. Since many updating methodologies are widely proposed both for MF-based and NN-based schemes, here we mainly focus on the embedding generation.


As implied in the title, we will use projection to accomplish the procedure above. Usually, MF-based methodologies factorize $\bm{M}$ into a product of two matrices, $\bm{B}_0 \bm{V}_0$, with $\bm{B}_0$ denotes the embedding matrix and $\bm{V}_0$, which is usually column orthogonal. For example, when using SVD to factorize $\bm{M}_0 = \bm{U\Sigma V}^\text{T}$, then $\bm{U}\sqrt{\bm{\Sigma}}$ is a frequently chosen embedding. Then, we can denote
\begin{equation}
    \bm{M}_0 = (\bm{U}\sqrt{\bm{\Sigma}})(\sqrt{\bm{\Sigma}} \bm{V}^{\text{T}}) = \bm{B}_0 \bm{V}_0,
\end{equation}
thus inversely, assuming $\bm{M}_0$ is a full rank matrix without loss of generality,
\begin{equation}
    \bm{B}_0 = \bm{M}_0\bm{V}_0^{-1} = \bm{M}_0 (\bm{V} \sqrt{\bm{\Sigma}}^{-1}),
\end{equation}
where target matrix $\bm{M}_0$ is projected on a ``basis matix'', $\bm{V}_0^{-1}$. Therefore, when a new node comes in and the new target matrix $\bm{M}_1 \in \mathbb{R}^{(n+1) \times (n+1)}$ is constructed, we can obtain $\boldsymbol{m} = \bm{M}_1\left[(n+1),0:n\right] \in \mathbb{R}^{1 \times n}$ as the target vector of the $i^{th}$ new node. After projecting it to each ``basis'', namely calculating the inner product of $\boldsymbol{m}$ and each column of the basis matrix, the results as ``coordinates'' on each basis, can be regard as the embedding. Therefore, its embedding can be constructed through

\begin{equation}
    \boldsymbol{b}_{new} = \boldsymbol{m}\bm{V}_0^{-1}.
\end{equation}
More interestingly, this projection process can also be treated as a weighted average of existed network embedding. Since $\bm{V}_0^{-1} = \bm{M}_0^{-1} \bm{B}_0$, thus
\begin{equation}
\label{projection}
    \boldsymbol{b}_{new} = \boldsymbol{m} \bm{M}_0^{-1} \bm{B}_0 = (m_1,m_2,...,m_n) \bm{M}_0^{-1} 
    \begin{bmatrix}
        \boldsymbol{b}_1 \\
        \boldsymbol{b}_2 \\
        \vdots \\
        \boldsymbol{b}_n \\
    \end{bmatrix},
\end{equation}
where $\boldsymbol{b}_i$ is the embedding vector of the $i^{th}$ new node. Therefore, the word ``projection'' can also mean using the correlation between the new node and first $n$ nodes for the linear combination to get the new embedding, whose time complexity is $O(n)$ with linear to the node size. In the streaming setting where nodes come orderly, the projection as Equation (\ref{projection}) can be applied to different $\boldsymbol{m}_i$ successively, or even parallelly when needed. Since using Theorem \ref{subspace} to judge when to restart is not time economical, thus practically, we can set a time interval or new node size for determing when to apply judgement.\par

As for the scene with discrete time interval, when $m$ nodes come in simultaneously, a more concise form of projection can be represented as

\begin{equation}
    \bm{B}_{new} = \hat{\bm{E}}^{(1)}\bm{V}_0^{-1},
\end{equation}
where $\hat{\bm{E}}^{(1)} = \bm{M}_1[(n+1):(n+m),0:n]$, the notation aligned with equation (\ref{2}). The columns of $\bm{V}_0^{-1}$, constructed through matrix decomposition, is mainly orthogonal, thus using Theorem \ref{subspace} to determine when to restart becomes more reasonable, since the stability of basis is the prerequisite of projection.\par

Apparently, SIP is a framework applicable for different static network embedding schemes. In Fig.\ref{SIP_image}, its flow chart is shown, with $T(\cdot)$ denoting the functions to convert adjacent matrix to target matrix. SIP can be divided into three main parts: static embedding, space drift judgement and embedding generation. The first part is used to construct the original basis, with the second part that confirms the feasibility of projection follows, which is not necessarily need for every generation epoch. For the embedding generation part, $\bm{M}_1$ is applied with $\bm{B}_0$ to construct new embeddings. 

\begin{figure}[htbp]
    \centering
    \includegraphics[width=\linewidth]{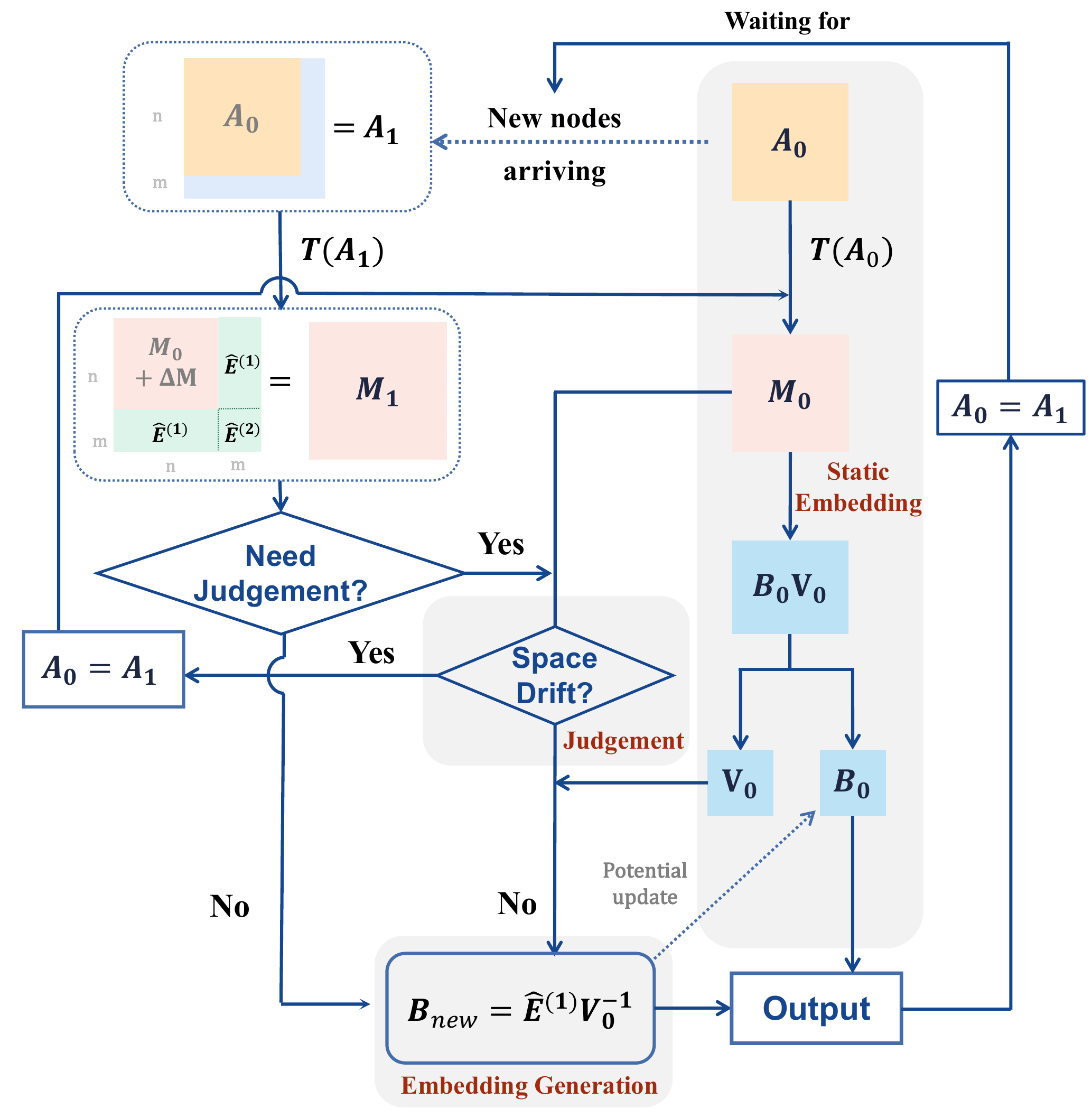}
    \caption{The flow chart of SIP framework for embedding generation. The framework mainly consists of three steps: at first, the static embedding is applied to first $n$ nodes to construct the basis matrix for later operation. Then, when new nodes arriving, a judgement of whether the space has drifted applied. When the original space is mainly maintained, new embedding generated through projection, otherwise the original embedding is retrained.}
    \label{SIP_image}
\end{figure}

\subsection{Instantiations of SIP}
\label{instantiation}

Even though SIP is direct and simple, when combining it with specific MF-based methodology, there are many variance based on the projection idea. In this subsection, four typical MF-based methodologies will be filled into the SIP framework, and the construction details will be elaborated to demonstrate the flexibility of SIP.

\subsubsection{SIP-LE}
The first instantiation goes into the classic method Laplacian Eigenmaps (LE) \cite{belkin2001laplacian}. The embedding acquired through LE is directly calculated through a generalized SVD problem $\bm{L}\boldsymbol{\alpha} = \lambda \bm{D} \boldsymbol{\alpha}$. While practically, LE is conducted through the eigen decomposition of the normalized laplacian matrix \cite{rozemberczki2020karate}, thus the normalized laplacian matrix is regarded as $\bm{M}_0$ here. The first $d$ eigenvectors acquired through this decomposition is regarded as $\bm{B}_0$, i.e.
\begin{equation}
    \bm{M}_0 \approx \bm{U}_d (\bm{\Sigma}_d \bm{U}_d^{\text{T}}) = \bm{B}_0 \bm{V}_0,
\end{equation}
thus
\begin{equation}
    \bm{V}_0^{-1} = \bm{U}_d \bm{\Sigma}_d^{-1}.
\end{equation}
Thereafter, the new embedding to be generated is
\begin{equation}
    \boldsymbol{b}_{new} = \boldsymbol{m}\bm{U}_d \bm{\Sigma}_d^{-1},
\end{equation}
where $\boldsymbol{m}$ is its target vector.

A parameter needs to be  introduced here is $C$, the space drift check state, which is a binary indicator to represent whether the judgment in SIP should be down. Then the procedure of SIP-LE can be represented as Algorithm \ref{SIP-LE-ALG}. Another trick in practice is that, in order to reduce the numerical instability introduced by small elements in $\bm{\Sigma}_d$, we replace all eigen-values that less than $10^{-1}$ as $1$ before calculating $\bm{\Sigma}_d^{-1}$.

\begin{algorithm}
\caption{SIP-LE} 
\label{SIP-LE-ALG}
\begin{algorithmic}
\REQUIRE Original factorization results $\bm{U}_d$ and $\bm{\Sigma}_d$, space drift check state $C$, new adjacent matrix $\bm{A}_1 \in \mathbb{R}^{(n+m) \times (n+m)}$.
\ENSURE Embedding for new nodes $\bm{B}_{new}$.
\STATE $\bm{M}_1 \Leftarrow$ Normalized laplacian matrix of $\bm{A}_1$ 
\IF{$C = 0$ or ($C=1$ and equation (\ref{judgement_equation}) satisfied)} 
    \STATE $\bm{B}_{new} = \bm{M}_1\left[n+1:n+m,0:n\right] \bm{U}_d \bm{\Sigma}_d^{-1}$
    \STATE Return $\bm{B}_{new}$
\ELSE
    \STATE $\left[\bm{U}, \bm{\Sigma}\right]$ = eigen-decomposition($\bm{M}_1$)
    \STATE $\bm{B}_{new} = \bm{U}\left[n+1:n+m,0:d\right]$
    \STATE Return $\bm{B}_{new}$
\ENDIF 
\end{algorithmic}
\end{algorithm}

\subsubsection{SIP-AROPE}
As for AROPE, this instantiation is quite similar to LE, since its core idea is also to eigen-decompose the target matrix $\bm{M} = T(\bm{A}) = w_1\bm{A} + w_2\bm{A}^2 +...+w_q\bm{A}^q$. Thereafter, its decomposition $SVD(\bm{M}) \approx \bm{U}_d \bm{\Sigma} \bm{V}_d^{\text{T}}$ are acquired, the embedding $\bm{B}$ is $\bm{U}_d \sqrt{\bm{\Sigma}_d}$. Therefore,
\begin{equation}
    \bm{V}_0 = \sqrt{\bm{\Sigma}_d}\bm{V}_d^{\text{T}}, 
\end{equation}
then two main important matrices for SIP, $\bm{M}$ and $\bm{V}_0$, are accessible.\par

\begin{algorithm}
\caption{SIP-AROPE} 
\label{SIP-AROPE-ALG}
\begin{algorithmic}
\REQUIRE Original factorization results $\bm{V}_d$ and $\bm{\Sigma}_d$, space drift check state $C$, new adjacent matrix $\bm{A}_1 \in \mathbb{R}^{(n+m) \times (n+m)}$.
\ENSURE Embedding for new nodes $\bm{B}_{new}$.
\STATE $\bm{M}_1 \Leftarrow T(\bm{A}_1) = w_1\bm{A}_1 + w_2\bm{A}_1^2 +...+w_q\bm{A}_1^q$ 
\IF{$C = 0$ or ($C=1$ and equation (\ref{judgement_equation}) satisfied)} 
    \STATE $\bm{B}_{new} = \bm{M}_1\left[n+1:n+m,0:n\right] \bm{V}_d\bm{\Sigma}_d^{-\frac{1}{2}}$
    \STATE Return $\bm{B}_{new}$
\ELSE
    \STATE $\left[\bm{X}, \bm{\Lambda}\right]$ = eigen-decomposition($\bm{M}_1$)
    \STATE Calculate the svd results $\left[ \bm{U}, \bm{\Sigma}, \bm{V} \right]$ using $\left[\bm{X}, \bm{\Lambda}\right]$ according to AROPE
    \STATE $\bm{B} = \bm{U} \sqrt{\bm{\Sigma}}$
    \STATE $\bm{B}_{new} = \bm{B}\left[n+1:n+m,0:d\right]$ 
    \STATE Return $\bm{B}_{new}$
\ENDIF 
\end{algorithmic}
\end{algorithm}

It can be seen from Algorithm \ref{SIP-AROPE-ALG}, the main procedure is the same as Algorithm \ref{SIP-LE-ALG}. The differences appeared mainly in the construction of the target matrix.

\subsubsection{SIP-GraRep}

The embedding generation in GraRep is different from LE and AROPE. Unlike integrally operate the target matrix to construct embedding, GraRep decomposes the first to the $k^{th}$ order of the positive log probability matrix and horizontally stack the embedding together to construct the final embedding. Therefore, when using SIP to generate new embedding, we are supposed to separately project different target vectors on different basis. \par

To be specific, computing $\bm{S}^k = (\bm{D}^{-1}\bm{A})^k$ and the $k^{th}$ positive log probability matrix $\bm{X}^k$ that satisfies $\bm{X}^k_{i,j} = log(\frac{\bm{S}^k_{i,j}}{\tau_j^k}) - log(\beta)$ with $\tau_j^k = \sum_p \bm{A}_{p,j}^k$ and $\beta$ is a parameter. After assigning negative entries of $\bm{X}^k$ to $0$ and $SVD(\bm{X}^k) \approx  \bm{U}_d^k \bm{\Sigma}_d^k \bm{V}_d^{k\text{T}}$, its SVD results $\bm{U}_d^k \sqrt{\bm{\Sigma}^k_d}$ are used for the $k-$step representation. Therefore, for each $\bm{X}^k$, its
\begin{equation}
    (\bm{V}_0^k)^{-1} = \bm{V}_d^k (\bm{\Sigma}_d^k)^{-\frac{1}{2}}.
\end{equation}
Details of this algorithm are presented in Algorithm \ref{SIP-GraRep-ALG}.

\begin{algorithm}
\caption{SIP-GraRep} 
\label{SIP-GraRep-ALG}
\begin{algorithmic}
\REQUIRE Original factorization results $\bm{V}_d^1,...,\bm{V}_d^k$ and $\bm{\Sigma}_d^1,...,\bm{\Sigma}_d^k$, low rank $d$, maximum order $k$, space drift check state $C$, new adjacent matrix $\bm{A}_1 \in \mathbb{R}^{(n+m) \times (n+m)}$.
\ENSURE Embedding for new nodes $\bm{B}_{new}$.
\IF{$C = 0$ or ($C=1$ and equation (\ref{judgement_equation}) satisfied)} 
    \STATE Create a list $\bm{B}_{list}$ that is able to add elements
    \FOR{i = 1 to $k$}
        \STATE $\bm{M}_1 \Leftarrow$ Construct $\bm{X}^k$ using $\bm{A}_1$
        \STATE $\bm{M}_{target} = \bm{M}_1\left[n+1:n+m,0:n\right]$
        \STATE $\bm{B}_{temp} = \bm{M}_{target} \bm{V}_d^k (\bm{\Sigma}_d^k)^{-\frac{1}{2}}$
        \STATE Add $\bm{B}_{temp}$ to the $\bm{B}_{list}$
    \ENDFOR
    \STATE Horizontally stack matrices in $\bm{B}_{list}$ to construct $\bm{B}_{new}$
    \STATE Return $\bm{B}_{new}$
\ELSE 
    \STATE $\bm{B}_{all} \Leftarrow$ Using GraRep to retrain embedding of $\bm{A}_1$
    \STATE $\bm{B}_{new} = \bm{B}_{all}\left[n+1:n+m, 0:d \right]$
    \STATE Return $\bm{B}_{new}$
\ENDIF 
\end{algorithmic}
\end{algorithm}

\subsubsection{SIP-NetMF}
Filling NetMF into the SIP framework is relatively the most complicated instantiation in this section, since NetMF algorithm consists of two decomposition steps. To be specific, the four main procedures of NetMF for a large window size \cite{qiu2018network} are: 
\begin{itemize}
    \item [(1)] Eigen-decompose $\bm{D}_0^{-\frac{1}{2}}\bm{A}_0\bm{D}_0^{-\frac{1}{2}}$ with its first $h$ eigenvectors and eigenvalues $\bm{U}_h \bm{\Lambda}_h \bm{U}_h^{\text{T}}$;
    \item [(2)] Construct an approximation matrix\\
    $\hat{\bm{M}}_0 = \frac{vol(\mathcal{G})}{b}\bm{D}_0^{-\frac{1}{2}}\bm{U}_h(\frac{1}{T}\sum_{r=1}^{\text{T}}\bm{\Lambda}_h^r)\bm{U}_h^{\text{T}}\bm{D}_0^{-\frac{1}{2}}$, where $vol(\mathcal{G}) = \sum_i\sum_j \bm{A}_0\left[i,j\right]$ is the volume of graph $\mathcal{G}$ and $T \& b$ are the context window size and the number of negative sampling in skip-gram respectively;
    \item [(3)] Compute $\hat{\bm{M}}^{'}_0 = max(\hat{\bm{M}}_0,1)$ to avoid numerical instability;
    \item [(4)] Approximate $log \hat{\bm{M}}^{'}_0 = \bm{U}_d \bm{\Sigma}_d \bm{V}_d^{T}$ using SVD, and return $\bm{U}_d \sqrt{\bm{\Sigma}_d}$ as network embedding.
\end{itemize}
Therefore, when a new node arrives with its connection information appears in $\bm{A}$, we are supposed to firstly generate its presentation corresponding to $\bm{U}_h$, then use it as a part to construct the new row vectors of $\hat{\bm{M}}_0^{'}$, then project this new target vector on the original basis. The deduction can be formulized as follows.\par

At first, denoting $\bm{H}_0 = \bm{D}_0^{-\frac{1}{2}}\bm{A}_0\bm{D}_0^{-\frac{1}{2}} \approx \bm{U}_h \bm{\Lambda}_h \bm{U}_h^{\text{T}}$, then, with a new node arriving, it becomes
\begin{equation}
    \bm{H}_1 = \left[
    \begin{matrix}
        \bm{H}_0 && \boldsymbol{\beta}^\text{T} \\
        \boldsymbol{\beta} && 0
    \end{matrix}
    \right] \approx \left[
    \begin{matrix}
        \bm{U}_h \\
        \boldsymbol{u}
    \end{matrix}
    \right] \bm{\Lambda}_h \left[
    \begin{matrix}
        \bm{U}_h^\text{T} && \boldsymbol{u}^\text{T} \\
    \end{matrix}  \right],
\end{equation}
if we recognize that the singular values kept unchanged. Thereafter, $\boldsymbol{u}\bm{\Lambda}_h \bm{U}_h^{\text{T}} = \boldsymbol{\beta}$ are supposed to be established. Thus $\boldsymbol{u} = \boldsymbol{\beta}\bm{U}_h\bm{\Lambda}_h^{-1}$ is the approximated vector, which can also be regarded as projecting $\boldsymbol{\beta}$ on the scaled basis $\bm{U}_h \bm{\Lambda}_h^{-1}$.\par

For the subsequent step, the target vector can be calculated follow the original equation as
\begin{equation}
    \hat{\boldsymbol{m}_0^{'}} = \frac{vol(\mathcal{G})}{b}d^{-\frac{1}{2}}\boldsymbol{u}(\frac{1}{T}\sum_{r=1}^{\text{T}}\bm{\Lambda}_h^r)\bm{U}_h^{\text{T}}\bm{D}_0^{-\frac{1}{2}},
\end{equation}
where $d$ is the degree between the new node and first $n$ vertices. After
\begin{equation}
    \boldsymbol{m} = log(max(\hat{\boldsymbol{m}}_0^{'},1))
\end{equation}
obtained, it can be project to the basis matrix $\bm{V}_0^{-1} = \bm{V}_d^{\text{T}} (\bm{\Sigma}_d)^{-\frac{1}{2}}$ to get the new embedding
\begin{equation}
    \boldsymbol{b}_{new} = \boldsymbol{m}\bm{B}_0 = \boldsymbol{m}\bm{V}_d(\bm{\Sigma}_d)^{-\frac{1}{2}}.
\end{equation}
It is also feasible when nodes come in group, under which circumstance we only need to treat vector $\boldsymbol{\beta}$ as a matrix. The SIP-NetMF algorithm that not constraint to one node is presented as follows.

\begin{algorithm}
\caption{SIP-NetMF} 
\label{SIP-NetMF-ALG}
\begin{algorithmic}
\REQUIRE Original parameter $vol(\mathcal{G}) \& b$, original eigen decomposition results $\bm{U}_h,\ \bm{\Lambda}_h, \ \bm{V}_d, \ \bm{\Sigma}_d$, space drift check state $C$, new adjacent matrix $\bm{A}_1 \in \mathbb{R}^{(n+m) \times (n+m)}$.
\ENSURE Embedding for new nodes $\bm{B}_{new}$.
\IF{$C = 0$ or ($C=1$ and equation (\ref{judgement_equation}) satisfied)}
    \STATE Let $\bm{D}_0$ and $\bm{D}$ be the degree diagonal matrix of first $n$ matrix and $m$ new nodes respectively
    \STATE Let $\bm{A}_{new} = \bm{A}_1\left[(n+1):(n+m),0:n \right]$ 
    \STATE Compute \\
    $\hat{\bm{M}} = \frac{vol(\mathcal{G})}{b} \bm{D}^{-1} \bm{A}_{new} \bm{D}_0^{-\frac{1}{2}} \bm{U}_h \bm{\Lambda}_h^{-1}(\frac{1}{T}\sum_{r=1}^{\text{T}}\bm{\Lambda}_h^r)\bm{U}_h^{\text{T}}\bm{D}_0^{-\frac{1}{2}}$
    \STATE $\hat{\bm{M}} = log(max(\hat{\bm{M}}),1)$
    \STATE Return $\bm{B}_{new} = \hat{\bm{M}} \bm{V}_d \bm{\Sigma}_d^{-\frac{1}{2}}$
\ELSE 
    \STATE Using NetMF to retrain embedding for $\bm{A}_1$
    \STATE Return $\bm{B}_{new}$
\ENDIF 
\end{algorithmic}
\end{algorithm}

\section{Experiment}
\subsection{Experimental Settings}

\begin{figure*}[htbp]
    \centering
    \includegraphics[width=\linewidth]{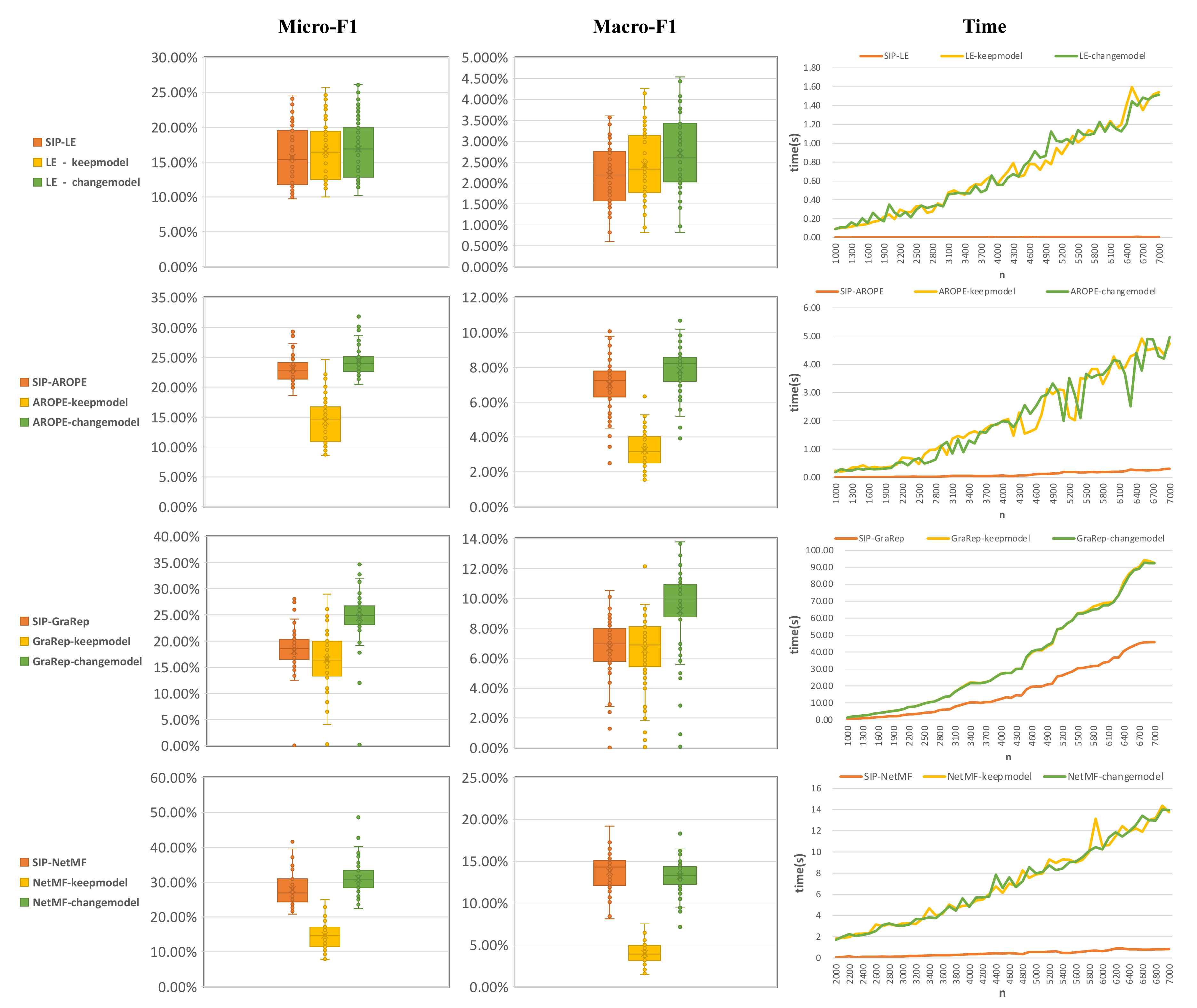}
    \caption{Classification performance of new arriving nodes on blogcatalog dataset, using LE, AROPE, GraRep and NetMF as well as their combination with SIP. For each method, the micro-F1 and macro-F1 results of different $n$ are integrated to draw box plots. The time for embedding generation is presented as line charts.}
    \label{experiment_fig}
\end{figure*}

In this study, as outlined in Section \ref{empirical restart}, we utilized three datasets that contain node labels to assess the performance of different SIP instantiations. The first dataset, Blogcatalog, is a social network that captures the relationships between bloggers and their corresponding interests are used as labels. The second dataset, PPI, represents a protein sub-network of Homo Sapiens, where the labels denote the biological states of each protein. Finally, Flickr is also a social network that reflects user contacts in Flickr with their interests serving as labels, where the first $20,000$ nodes are used. After removing nodes that were not part of the giant connected component in all three datasets, we were left with $9631$, $2591$, and $18813$ nodes for Blogcatalog, PPI, and Flickr respectively. To evaluate the efficacy of SIP on these datasets, we combined four static embedding methodologies- LE, AROPE, GraRep, and NetMF- with SIP.

This paper focus on the node classification task, aiming to quantify the ability of SIP to construct embeddings for new arriving nodes. Denoting the original network as $\mathcal{G}$, then, different initial nodes sizes $n$ are chosen to construct $\mathcal{G}_0$ for the static embedding, and the logistic classification is chosen as the downstream model. For the PPI dataset, $n$ varies from $500$ to $2500$, while the other two datasets choose $n$ from $1000$ to $7000$, with increment step equals to $100$ for all of them.\par

For each $\mathcal{G}$ and $n$, $m_0$ is calculated and three verification modes are applied. The first mode uses the embedding generated by SIP and keeps the models unchanged. The second mode generates $m_0$ new embeddings by applying the corresponding static embedding to $\mathcal{G}_1$ with $n+m_0$ nodes, with the models unchanged as well. The third mode involves retraining both the embedding and downstream model, which is considered as the ground truth result for the node classification effect. It is worth noting that for NetMF, $m_0$ is relatively small, and it equals zero for $n < 1000$ in the PPI dataset and for all $n$ in the Flickr dataset. Therefore, the validation of SIP-NetMF in PPI starts from $n = 1000$, while in Flickr, $m_0$ is calculated using $\bm{A}$ as the target matrix.\par

The evaluation metrics used in this study are micro-F1 and macro-F1, which are used to reflect the classification results. The running time for embedding acquisition is also taken into account to indicate time complexity. All the hyperparameters within the methods are consistent with the code released by each original paper.\par

\subsection{Experimental results}

\begin{table*}[h]
\renewcommand\arraystretch{2}
    \caption{The running time of generating embedding for new nodes, as well as the classification performance for these nodes are presented. 
    SIP combined with LE, AROPE, GraRep and NetMF are presented, along with the original performance of these methods under the pre-trained model and the retrained model.}
    \centering
    \begin{tabular}{p{1.6cm}|p{1.45cm}  p{1.45cm}  p{1.45cm}|p{1.45cm}  p{1.45cm}  p{1.45cm} |p{1cm} p{1cm}  p{1cm} }
    
     \hline
      & \multicolumn{3}{|c|}{\textbf{Micro - F1 (\%)}} & \multicolumn{3}{|c|}{\textbf{Macro - F1 (\%)}}  & \multicolumn{3}{|c}{\textbf{Time (s)}}\\
     \hline\hline
     \rule{0pt}{7pt}
      & Blog & PPI & Flickr & Blog & PPI & Flickr & Blog & PPI & Flickr \\
      \hline
      \textbf{SIP-LE }&$15.72\pm0.04 $&	$5.68\pm0.04 $&	$18.84 \pm0.06 $&	$2.19\pm0.01 $&	$1.17\pm0.01 $&	$0.50\pm0.00 $&	$0.00 $&	$0.00 $&	$0.00 $\\
      \textbf{LE-keepmodel}&$16.53 \pm0.04 $&	$7.96 \pm0.03 $&	$19.26 \pm0.06 $&	$2.44 \pm0.01 $&	$1.73 \pm0.01 $&	$0.55 \pm0.00 $&	$0.75$&	$0.16 $&	$0.44$ \\
      \textbf{LE-changemodel}&$16.88 \pm0.04 $&	$8.95 \pm0.03 $&	$19.48 \pm0.06 $&	$2.73 \pm0.01 $&	$2.30 \pm0.01 $&	$0.60 \pm0.00 $&	$0.75$&	$0.14$&	$0.45$ \\
     \hline
     \textbf{Ratio}&$\textbf{93.16 \%}$&	$\textbf{63.38 \%}$&	$\textbf{96.73 \%}$&	$\textbf{80.33 \%}$&	$\textbf{50.92 \%}$&	$\textbf{83.26 \%}$&	$\textbf{0.43 \%}$&	$\textbf{0.66 \%}$&	$\textbf{0.38 \%}$\\
     \hline\hline
    \textbf{SIP-AROPE}& $23.11 \pm2.34 $&	$14.84 \pm3.87 $&	$26.83 \pm3.60 $&	$7.11 \pm1.41 $&	$7.96 \pm2.12 $&	$3.74 \pm1.22 $&	$0.12$&	$0.01$&	$0.05$\\
    \textbf{AROPE-keepmodel}& $14.30 \pm3.72 $&	$8.25 \pm2.93 $&	$19.43 \pm3.45 $&	$3.19 \pm1.04 $&	$3.71 \pm1.92 $&	$1.34 \pm0.80 $&	$2.27$&	$0.37$&	$2.07 $\\
    \textbf{AROPE-changemodel}& $24.43 \pm2.40 $&	$14.94 \pm3.77 $&	$27.18 \pm3.44 $&	$7.93 \pm1.26 $&	$8.38 \pm2.08 $&	$3.99 \pm1.19 $&	$2.23$&	$0.34$&	$2.11$\\
    \hline
    \textbf{Ratio}&$\textbf{94.60 \%}$&	$\textbf{99.33 \%}$&	$\textbf{98.70 \%}$&	$\textbf{89.69 \%}$&	$\textbf{95.02 \%}$&	$\textbf{93.71 \%}$&	$\textbf{5.19 \%}$&	$\textbf{3.09 \%}$&	$\textbf{2.60 \%}$\\
     \hline\hline
     \textbf{SIP-GraRep}& $18.58 \pm3.41 $&	$12.87 \pm5.18 $&	$22.57 \pm4.75 $&	$6.80 \pm1.74 $&	$7.40 \pm3.07 $&	$5.22 \pm1.56 $&	$22.35 $&	$1.73 $&	$17.72 $\\
     \textbf{GraRep-keepmodel}& $16.04 \pm4.73 $&	$9.88 \pm5.55 $&	$5.75 \pm4.29 $&	$6.54 \pm2.19 $&	$5.42 \pm3.10 $&	$1.46 \pm1.26 $&	$53.64 $&	$4.62 $&	$39.04 $\\
     \textbf{GraReP-changemodel}&$25.35 \pm3.28 $&	$14.65 \pm4.24 $&	$24.83 \pm4.22 $&	$9.56 \pm2.28 $&	$8.70 \pm2.95 $&	$6.30 \pm1.64 $&	$53.01 $&	$4.59 $&	$39.16 $\\
     \hline
     \textbf{Ratio}& $\textbf{73.28 \%}$&	$\textbf{87.87 \%}$&	$\textbf{90.90 \%}$&	$\textbf{71.12 \%}$&	$\textbf{85.13 \%}$&	$\textbf{82.87 \%}$&	$\textbf{42.16 \%}$&	$\textbf{37.73 \%}$&	$\textbf{45.24 \%}$\\
     \hline\hline
     \textbf{SIP-NetMF}& $27.84 \pm4.48 $&	$19.56 \pm6.72 $&	$26.12 \pm7.04 $&	$13.87 \pm1.96 $&	$11.12 \pm5.02 $&	$3.40 \pm1.67 $&	$0.48$&	$0.11$&	$0.20$\\
     \textbf{NetMF-keepmodel}& $14.74 \pm3.96 $&	$7.51 \pm5.05 $&	$12.69 \pm6.25 $&	$4.09 \pm1.38 $&	$4.06 \pm2.97 $&	$0.82 \pm0.68 $&	$7.72$&	$2.07$&	$7.09 $\\
     \textbf{NetMF-changemodel}&  $31.09 \pm4.81 $&	$22.33 \pm4.20 $&	$28.06 \pm6.69 $&	$13.31 \pm1.62 $&	$12.55 \pm4.25 $&	$3.20 \pm1.63 $&	$7.68$&	$1.98$&	$6.98$\\
     \hline
     \textbf{Ratio}& $\textbf{89.55 \%}$&	$\textbf{87.61 \%}$&	$\textbf{93.09 \%}$&	$\textbf{104.19 \%}$&	$\textbf{88.60 \%}$&	$\textbf{106.22\%}$&	$\textbf{6.29 \%}$&	$\textbf{5.53 \%}$&	$\textbf{2.85 \%}$\\
     \hline\hline 
     \multicolumn{10}{c}{}
    \end{tabular}
    
    \label{experiment table}
\end{table*}
Fig.\ref{experiment_fig} is drawn based on the experimental results of blogcatalog dataset, which gives an intuitive illustration of the SIP effect. Based on the classification results under different $n$, the box plots of micro-F1 and macro-F1 using different methods are presented, and the time spent on embedding acquirement with the increase of $n$ are plotted in line graphs. As for the results, firstly, nearly all four SIP results are superior than the keep model scheme, whose new embeddings are acquired through retraining. An exception is SIP-LE, while under this circumstance SIP results are still close to the ground truth, thus it is still acceptable considering the much higher time efficiency. Therefore, in streaming setting, when the downstream model tends to be constant, i.e. considering the difficulty or infeasibility of model adjustment, SIP is a fast and accurate approach to obtain new embeddings compared with retraining. Moreover, even when the downstream model is adjusted or even retrained, in most cases classification results using SIP is still close to the ground truth, while the time consumption is much lower than retraining. In this experiment, the time used for retraining model is not take into account since fitting a logistic classification is efficient under this data size. While in real applications where millions of samples are considered, and millions of parameters are involved, the time saved for model revising is nonnegligible.\par

More detailed numerical results are presented in Table \ref{experiment table}, where all results for four methods in three datasets are reported. For micro-F1 and macro-F1, their averaged performance and standard errors are shown. We also add a new creteria ``\textbf{Ratio}'', which reflects the performance of SIP compared with the ground truth. For example, for micro-F1 metirc of LE, ratio equals to the micro-F1 result of SIP-LE divided by the result of LE. Therefore, a high ratio for classification performance and a low ratio for time is desired. The results show that for all four instantiations, SIP combined with them are able to achieve approximately $90 \%$ classification accuracy compared with the ground truth effect, while the running time is much lower. Even for SIP-GraRep, the most time consuming scheme, it still only require less than half time to obtain up to $90 \%$ accuracy. Therefore, not only the efficacy and efficiency, but also the general adaptation of SIP are empirically proved.

\section{Conclusion}
In this paper, we tackle the challenge of quickly generating embeddings when new vertices are added to networks. We begin by examining the changes in the embedding space that occur as a result of changes in the network structure, and propose a method for measuring this "space drift." With this method, we can determine the threshold number of nodes at which both the embedding and downstream models must be restarted for a specific network. Within this threshold, we introduce the SIP framework for generating new embeddings, which utilizes four initialization methods: LE, AROPE, GraRep, and NetMF. The time complexity of the SIP framework is linear in relation to network size. We perform empirical studies on three datasets, which demonstrate the superiority of SIP for node classification tasks.\par

Due to space constraints, this paper does not provide a comprehensive discussion of all details pertaining to the embedding generation procedure. Specifically, although the algorithms presented in this paper utilize matrix multiplication for embedding generation, we briefly highlighted the potential of parallel computation as a more efficient alternative, especially in the streaming setting. Additionally, while $m_0$, as proposed in this paper, serves as an adequate restart timing for both the basis and downstream models, our empirical findings suggest that embedding generation is still superior even when $m$, the number of new nodes, slightly exceeds $m_0$. Consequently, short-term continuous generation remains a practical option when equation (\ref{judgement_equation}) no longer satisfies the aforementioned conditions.\par

In future work, SIP could potentially be integrated with node update methods to accommodate various scenarios, including link prediction, among others. Furthermore, the classifier can be modified for various downstream tasks.

\section*{Acknowledgment}

The authors are grateful for the financial support from the National Natural Science Foundation of China (Grant Nos. 72021001 and 71871006).

\ifCLASSOPTIONcaptionsoff
  \newpage
\fi



\bibliographystyle{IEEEtran}
\bibliography{ieee}
%

%


\begin{IEEEbiography}[{\includegraphics[width=1in,height=1.25in,clip,keepaspectratio]{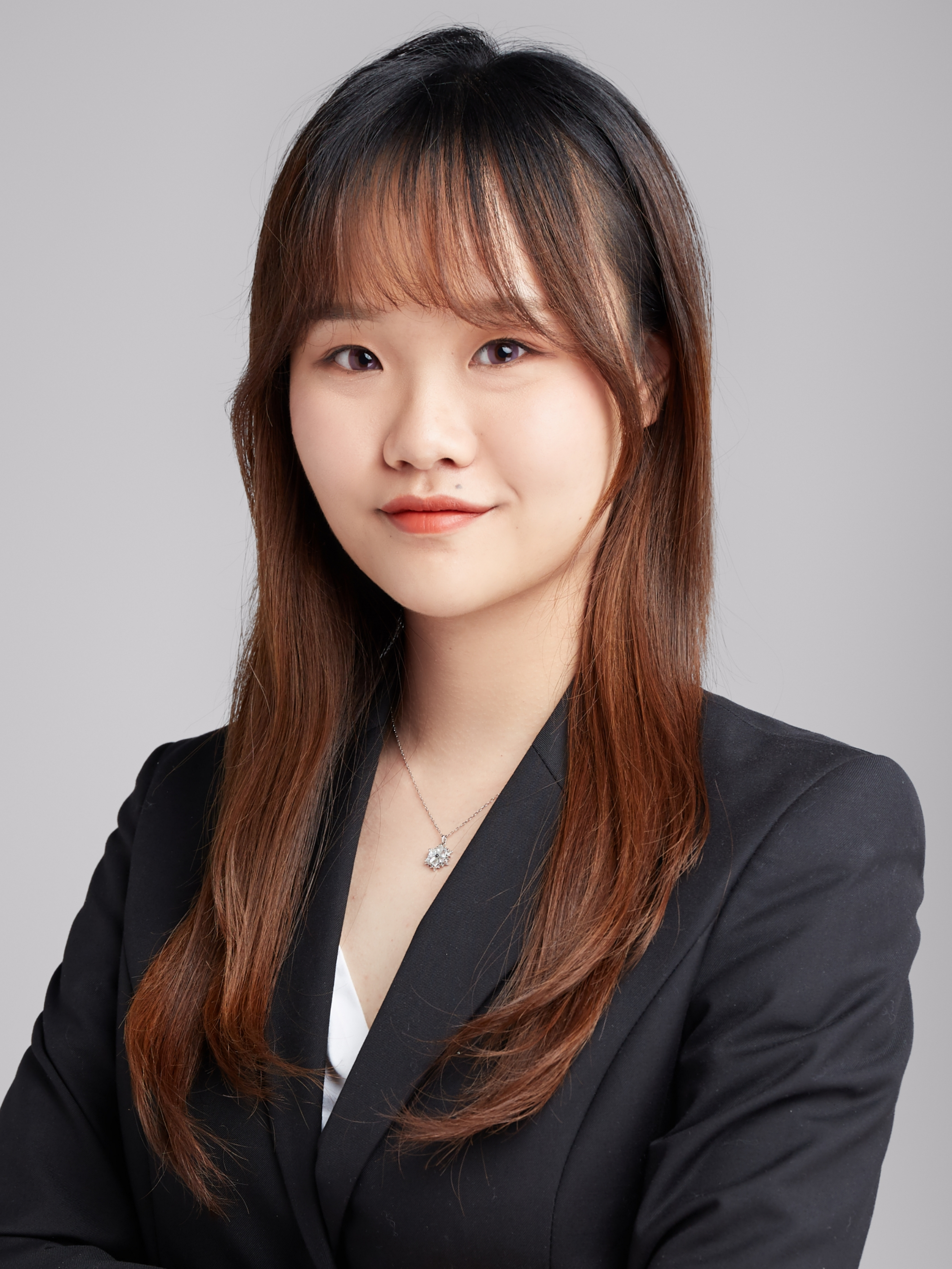}}]{Yanwen Zhang}
Yanwen Zhang received her B.E. degree from Beihang University in 2020. She is currently pursuing the Ph.D degree with Economics and Management, Beihang University. Her research interests include matrix theory and its applications in network analysis.
\end{IEEEbiography}

\begin{IEEEbiography}[{\includegraphics[width=1in,height=1.25in,clip,keepaspectratio]{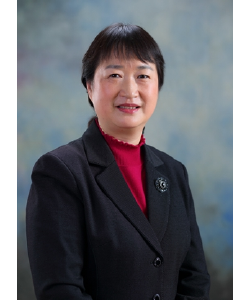}}]{Huiwen Wang}
Huiwen Wang received the B.S. degree from Beihang University in 1982, DEA from Paris Dauphine University in 1989 and Ph.D degree from Beihang University in 1992. She is currently a professor at the School of Economics and Management, Beihang University.
\end{IEEEbiography}

\begin{IEEEbiography}[{\includegraphics[width=1in,height=1.25in,clip,keepaspectratio]{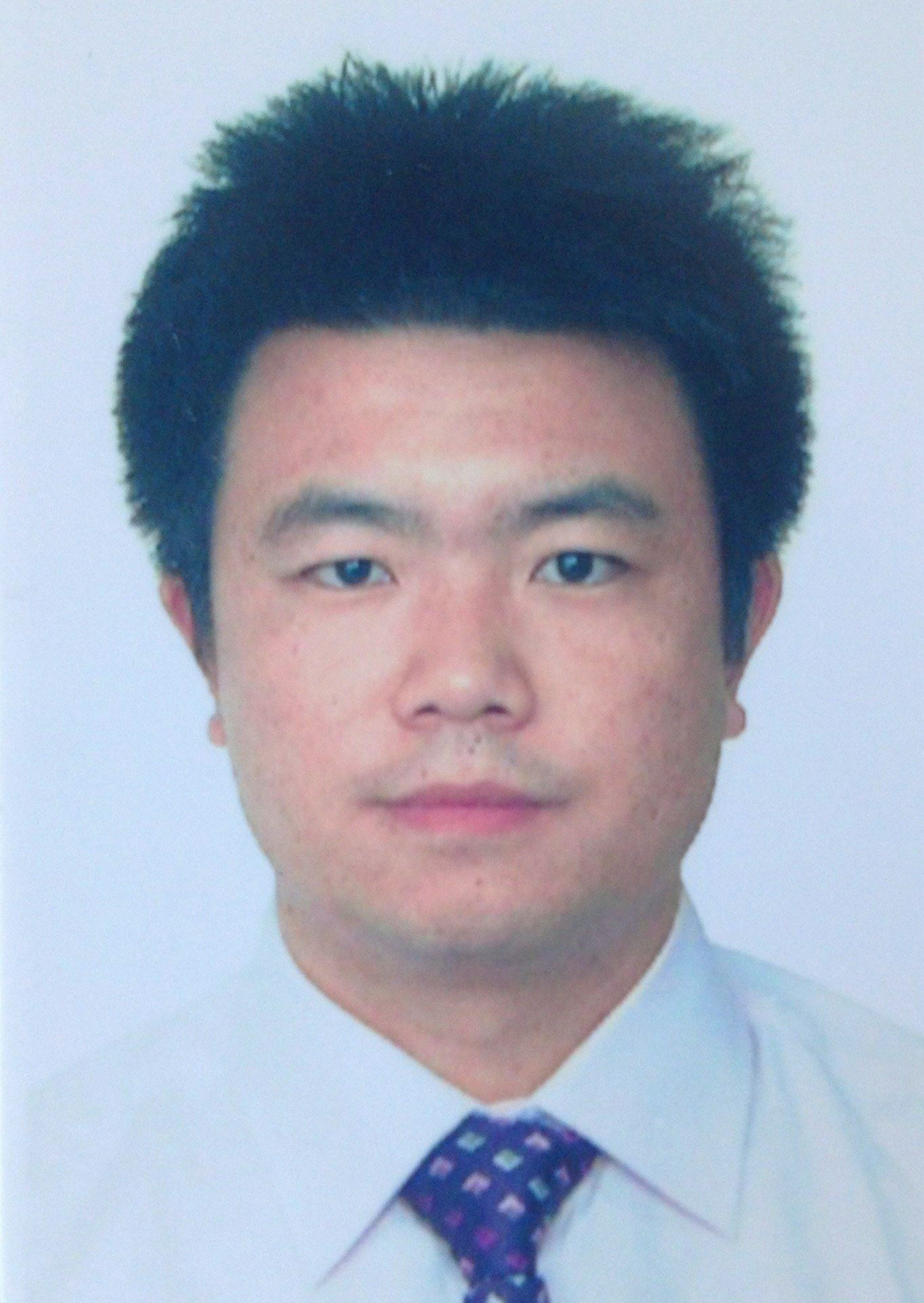}}]{Jichang Zhao}
Jichang Zhao received his B.E. and Ph.D. degrees from
Beihang University in 2008 and 2014. He is currently an associate professor at the School of Economics and Management, Beihang University. His research interests include computational social science and complex systems.
\end{IEEEbiography}





\end{document}